\documentclass[11pt]{article}
\usepackage{fullpage}
\usepackage{authblk}
\usepackage{graphicx}
\usepackage{amsmath}
\usepackage{amsfonts}
\usepackage{amssymb}
\usepackage{amsthm}
\usepackage{mathrsfs}
\usepackage{bm}
\usepackage{mathtools}
\usepackage{enumitem}

\usepackage{url}
\usepackage{hyperref}

\usepackage{caption}
\usepackage{subcaption}
\newtheorem{theorem}{Theorem}
\newtheorem{lemma}[theorem]{Lemma}

\newtheorem{corollary}[theorem]{Corollary}

\newtheorem*{claim*}{Claim}
\newtheorem*{remark*}{Remark}

\newtheorem*{definition*}{Definition}

\makeatletter
\renewcommand\section{%
  \@startsection{section}{1}
                {\z@}%
                {-3.5ex \@plus -1ex \@minus -.2ex}%
                {2.3ex \@plus.2ex}%
                {\large\bfseries}
}
\renewcommand\subsection{%
  \@startsection{subsection}{2}
                {\z@}%
                {-3.25ex\@plus -1ex \@minus -.2ex}%
                {1sp}
                {\normalsize\bfseries}
}
\renewcommand\subsubsection{%
  \@startsection{subsubsection}{3}
                {\z@}%
                {-3.25ex\@plus -1ex \@minus -.2ex}%
                {1sp}
                {\normalfont\normalsize}
}
\makeatother

\title{{\LARGE\bf  Population Monotonicity in Matching Games}
\footnote{Supported in part by National Natural Science Foundation of China (Nos.\,12001507, 11971447 and 11871442) and Natural Science Foundation of Shandong (No.\,ZR2020QA024).}}

\author{Han Xiao
and Qizhi Fang
}

\affil{School of Mathematical Sciences\\Ocean University of China\\Qingdao, China\\ \{hxiao,~qfang\}$@$ouc.edu.cn}

\date{}
\begin{document}


\maketitle

\openup 1.2\jot


\begin{abstract}
A matching game is a cooperative profit game defined on an edge-weighted graph,
where the players are the vertices and the profit of a coalition is the maximum weight of matchings in the subgraph induced by the coalition.
A population monotonic allocation scheme is a collection of rules defining how to share the profit among players in each coalition such that every player is better off when the coalition expands.
In this paper, we study matching games and provide a necessary and sufficient characterization for the existence of population monotonic allocation schemes.
Our characterization also implies that whether a matching game admits population monotonic allocation schemes can be determined efficiently.

\hfill



\noindent\textbf{JEL Classifcation: } C71, C78.
\hfill

\hfill
\end{abstract}

\section{Introduction}
Matching games, which capture matching markets with transferable utilities, make one of the cornerstones in cooperative game theory.
Roughly speaking, a matching game is a cooperative profit game defined on an edge-weighted graph,
where the players are the vertices and the profit of a coalition is the maximum weight of matchings in the subgraph induced by the coalition.
The following setting, taken form \cite{BKP12, EK01, Vazi21}, vividly illustrates the underlying scenario of matching games.
Consider a group of tennis players that will participate in a doubles tournament.
To represent the underlying structure, we introduce a weighted graph $G = (V, E;w)$.
The vertices are the players,
an edge $ij$ represents that players $i$ and $j$ are compatible doubles partners,
and $w: E\rightarrow \mathbb{R}_{+}$ is a function where $w_{ij}$ represents the expected prize money if $i$ and $j$ partner up in the tournament.
The total prize money for any subgroup $S\subseteq V$ of players in the doubles tournament is the maximum weight of matchings in the edge-weighted graph induced by $S$.
In particular, a matching game is called a \emph{simple matching game} when $w\equiv1$ and called an \emph{assignment game} when $G$ is bipartite, respectively.

An essential issue in a cooperative profit game is how to distribute the total profit among the players in a coalition.
There are many criteria for evaluating how ``good'' an allocation is, such as stability,  fairness, and satisfaction.
Emphases on different criteria lead to different allocation concepts, e.g., the core, the Shapley value, and the nucleolus.
Various allocation concepts have been studied extensively and intensively for matching games.

The core, which addresses the issue of stability for the grand coalition, concerns one of the most attractive allocations.
Shapley and Shubik \cite{SS71} show that the core of assignment games is always non-empty.
Deng et al. \cite{DIN99} provide a complete characterization for the core non-emptiness of matching games.
Eriksson and Karlander \cite{EK01} initialize the study on extreme core allocations for matching games.
N{\'u}{\~n}ez and Rafels \cite{NR03} achieve a complete characterization on extreme core allocations for assignment games.
Toda \cite{Toda05} propose an axiomatic characterization for the core of assignment games.
Klaus and Nichifor \cite{KN10} investigate the relation of the core with other allocation concepts for matching games.
Recently, Vazirani \cite{Vazi21} studies the approximate core and achieve the best possible approximation factor.

The nucleolus, which maximizes the minimum satisfaction among all players, is another well-studied allocation.
Solymosi and Raghavan \cite{SR94} propose an efficient algorithm for computing the nucleolus of assignment games.
Kern and Paulusma \cite{KP03} introduce an efficient algorithm for computing the nucleolus of simple matching games.
Llerena et al. \cite{LNR15} characterize the nucleolus by properties of the core and the kernel for assignment games.
Recently, K\"{o}nemann et al. \cite{KPT20} show that the nucleolus for matching games can be computed efficiently, which resolves an outstanding open problem proposed by Faigle et al. \cite{FKFH98}.

In addition to allocations, convexity is also a desirable property to study in cooperative profit games,
as convex games possess nice properties both economically and computationally.
However, Solymosi and Raghavan \cite{SR01} show that even assignment games are hardly convex.
Recently, Kumabe and Maehara turn to study generalizations of matching games and succeed in characterizing the convexity of $b$-matching games \cite{KM20-IJCAI} and hypergraph matching games \cite{KM20-AAMAS} respectively.

In this paper, we study population monotonic allocation schemes in matching games.
An allocation scheme is a collection of rules defining how to share the profit among players in every coalition.
An allocations scheme is population monotonic if every player is better off when the coalition expands.
Population monotonic allocation schemes possess an appealing snowball effect, i.e., the incentive to join a coalition increases as the coalition grows larger \cite{Spru90}.
Moreover, population monotonic allocation schemes yield group strategyproof mechanisms which resist collusion among players \cite{Moul99, MS01}.
However, very few results are known even for allocation schemes.
Deng et al. \cite{DINZ00} study allocation schemes consisting of core allocations for every coalition and achieve a sufficient characterization for simple matching games.
Immorlica et al. \cite{IMM08} study the limitation of approximate population monotonic allocations schemes and show that no constant approximation factor exists for matching games.
In this paper, we complete this line of research by providing a necessary and sufficient characterization for the existence of population monotonic allocation schemes in matching games.  

The remainder of this paper is organized as follows.
In Section \ref{sec:preliminaries}, some basic notions in cooperative game theory are reviewed.
Section \ref{sec:PM} is devoted to a complete characterization for the existence of population monotonic allocation schemes in matching games.
Section \ref{sec:disscussion} concludes the results in this paper and discusses the directions of future work.

\section{Preliminaries}
\label{sec:preliminaries}

A \emph{cooperative game} $\Gamma=(N,\gamma)$ consists of a \emph{player set} $N$ and a \emph{characteristic function} $\gamma:2^N\rightarrow \mathbb{R}$ with convention $\gamma(\emptyset)=0$.
We call $N$ the \emph{grand coalition} and call $S$ a \emph{coalition} for any $S\subseteq N$.

An \emph{allocation} of $\Gamma$ is a vector $\boldsymbol{x}=(x_i)_{i\in N}$ specifying how to distribute the profit among players in the grand coalition $N$.
A \emph{core allocation} is an allocation $\boldsymbol{x}=(x_i)_{i\in N}$ satisfying \emph{efficiency} and \emph{coalitional rationality} conditions,
\begin{enumerate}
  \item[\textendash] \emph{efficiency}: $\sum_{i\in N} x_{i}=\gamma(N)$;
  \item[\textendash] \emph{coalitional rationality}: $\sum_{i\in S}x_i
\geq \gamma(S)$ for any $S\subseteq N$.
\end{enumerate}
The \emph{core} is the set of all core allocations.

An \emph{allocation scheme} of $\Gamma$ is a collection of allocations $\mathcal{X}=(\boldsymbol{x}_S)_{S\in 2^N \backslash \{\emptyset\}}$ specifying how to distribute the profit among players in every nonempty coalition $S\subseteq N$.
A \emph{population monotonic allocation scheme} (\emph{PMAS} for short) is an allocation scheme $\mathcal{X}=(\boldsymbol{x}_S)_{S\in 2^N \backslash \{\emptyset\}}$ satisfying \emph{efficiency} and \emph{monotonicity} conditions,
\begin{enumerate}
  \item[\textendash] \emph{efficiency}: $\sum_{i\in S} x_{S,i}=\gamma(S)$ for any $S\in 2^N\backslash \{\emptyset\}$;
  \item[\textendash] \emph{monotonicity}: $x_{S,i}\leq x_{T,i}$ for any $S, T\in 2^N\backslash \{\emptyset\}$ with $S\subseteq T$ and any $i\in S$.
\end{enumerate}
We call $\Gamma$ \emph{population monotonic} if it admits a PMAS.
 
Now we define matching games.
Let $G=(V,E;w)$ be a graph with an edge weight function $w:E\rightarrow \mathbb{R}_+$.
Throughout this paper, we always assume that $w_e>0$ for any $e\in E$ since we may remove every edge $e$ with $w_e=0$ in any maximum weight matching.
The \emph{matching game} on $G=(V,E;w)$ is cooperative profit game $\Gamma_{G}=(N,\gamma)$, where $N=V$ and $\gamma (S)$ equals the maximum weight of matchings in the induced subgraph $G[S]$ for any $S\subseteq N$.

\section{Characterizing population monotonicity}
\label{sec:PM}

To characterize population monotonic matching games,
we first study properties of fundamental structures in Subsection \ref{subsec:structures},
then identify some forbidden structures in Subsection \ref{subsec:obstructions},
and finally achieve a necessary and sufficient characterization in Subsection \ref{subsec:characterization}.
In the remainder of this section, $G=(V,E;w)$ denotes a simple graph with weight function $w:E\rightarrow \mathbb{R}_+$,
and $\Gamma_G=(N,\gamma)$ denotes the matching game on $G$, unless stated otherwise. 

\subsection{Fundamental structures}
\label{subsec:structures}
This subsection is devoted to properties of fundamental structures in population monotonic matching games.
A graph is a \emph{complete graph} if every pair of vertices is connected by an edge.
A complete graph with $n$ vertices is denoted by $K_n$.
A graph is a \emph{path graph} if it is a tree with maximum degree no more than $2$.
A path graph with $n$ vertices is denoted by $P_n$.
A graph is a \emph{cycle graph} if it is a connected graph where every vertex has degree $2$.
A cycle graph with $n$ vertices is denoted by $C_n$.
A graph is a \emph{paw graph} if it is isomorphic to the third graph in Figure \ref{fig:structures}.
A graph is a \emph{diamond graph} if it is isomorphic to the last graph in Figure \ref{fig:structures}.
As we shall see, graphs in Figure \ref{fig:structures} are fundamental structures (induced subgraphs) in graphs inducing population monotonic matching games.
Thus we study properties of these structures first.

\begin{figure}[t]
\vspace{-1em}
\centerline{\includegraphics[width=1.5\textwidth]{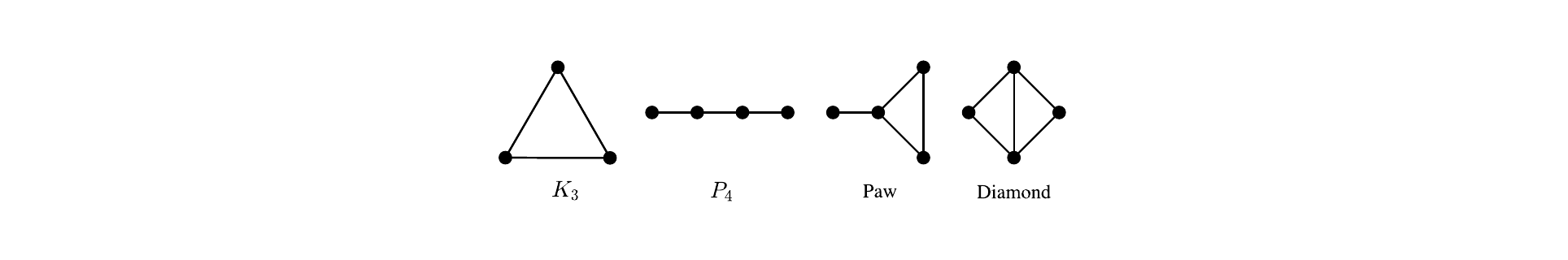}}
\vspace{-2em}
\caption{Fundamental structures}
\label{fig:structures}
\vspace{-1em}
\end{figure}

\begin{lemma}[$K_3$-property]
\label{thm:K3}
Let $H\subseteq G$ with $E(H)=\{12,13,23\}$ be an induced subgraph with $w_{12}\leq w_{13}\leq w_{23}$.
If $\Gamma_G $ is population monotonic, then $w_{23} \geq w_{12}+w_{13}$.
\end{lemma}

\begin{proof}
Let $\mathcal{X}=(\boldsymbol{x}_{S})_{S\in 2^N \backslash \{\emptyset\}}$ be a PMAS in $\Gamma_{G}$.
By efficiency and monotonicity of PMASes, we have
  \begin{equation}\label{eq:C3}
    \begin{split}
      2 w_{23}
      = & ~ 2 \gamma (\{1,2,3\})\\
       = & ~ ( x_{\{1,2,3\},1}+x_{\{1,2,3\},2}+x_{\{1,2,3\},3} ) + ( x_{\{1,2,3\},1}+x_{\{1,2,3\},2}+x_{\{1,2,3\},3} )\\
      \geq & ~ (x_{\{1,2\},1}+x_{\{1,2\},2}) + (x_{\{1,3\},3}+x_{\{1,3\},1}) + (x_{\{2,3\},2}+x_{\{2,3\},3})\\
      = & ~ \gamma(\{1,2\})+ \gamma(\{1,3\})+ \gamma(\{2,3\})\\
      = & ~ w_{12}+w_{13}+w_{23}.
    \end{split}
  \end{equation}
It follows that $w_{23}\geq w_{12}+w_{13}$.
\end{proof}

The following lemma provides a unified framework to develop more results on graphs inducing population monotonic matching games.
We remark that Lemma \ref{thm:cover} \emph{\ref{itm:2K3}} never occurs in population monotonic matching games, but it serves as an intermediate result in proving Lemma \ref{thm:K4}.

\begin{lemma}
\label{thm:cover}
Let $S=\{1,2,3,4\}$ be a vertex subset of $G$.
Let $H_1=G[\{1,2,3\}]$ and $H_2=G[\{2,3,4\}]$ be two induced subgraphs of $G$.
\begin{enumerate}[label={\emph{($\roman*$)}}]
	\item \label{itm:2P3} Assume that $E(H_1)=\{12,23\}$ and $E(H_2)=\{23,34\}$.
	If $\Gamma_G$ is population monotonic, then $w_{23} \geq w_{12}+w_{34}$.
	\item \label{itm:P3K3} Assume that $E(H_1)=\{12,23\}$, $E(H_2)=\{23,24,34\}$ and $w_{23}\geq w_{24}$.
	If $\Gamma_G$ is population monotonic, then $w_{23}\geq w_{12}+w_{34}$ and $w_{23}\geq w_{24}+w_{34}$.
	\item \label{itm:2K3} Assume that $E(H_1)=\{12,13,23\}$, $E(H_2)=\{23,24,34\}$ and $w_{12}\geq \max\{w_{13}, w_{23}\}$.
	If $\Gamma_G$ is population monotonic, then $w_{24}\geq w_{23}+w_{34}$.
\end{enumerate}
\end{lemma}
For the vertex subset $S$ and induced subgraphs $H_1$, $H_2$ in Lemma \ref{thm:cover}, we say that $H_1$ and $H_2$ make a \emph{$2 P_3$-cover} of $S$ in \emph{\ref{itm:2P3}},
a \emph{$(P_3, K_3)$-cover} of $S$ in \emph{\ref{itm:P3K3}},
and a \emph{$2 K_3$-cover} of $S$ in \emph{\ref{itm:2K3}}, respectively.

\begin{proof}
Let $\mathcal{X}=(\boldsymbol{x}_{S})_{S\in 2^N \backslash \{\emptyset\}}$ be a PMAS in $\Gamma_{G}$.
By efficiency and monotonicity of PMASes, we have
\begin{equation}
\label{eq:4cover}
    \begin{split}
          &~ \gamma (\{1,2,3\})+\gamma (\{2,3,4\})\\ 
       = &~ (x_{\{1,2,3\},1}+x_{\{1,2,3\},2}+x_{\{1,2,3\},3})+(x_{\{2,3,4\},2}+x_{\{2,3,4\},3}+x_{\{2,3,4\},4})\\
      \geq &~ (x_{\{1,2\},1}+x_{\{1,2\},2}) + (x_{\{2,3\},3}+x_{\{2,3\},2}) + (x_{\{3,4\},3}+x_{\{3,4\},4}) \\
      = &~ \gamma(\{1,2\})+ \gamma(\{2,3\})+ \gamma(\{3,4\})\\
      = &~ w_{12}+w_{23}+w_{34}.
   \end{split}
\end{equation}

\emph{\ref{itm:2P3}}
Assume that $E(H_1)=\{12,23\}$ and $E(H_2)=\{23,34\}$.
Thus
\begin{equation}
\label{eq:P4}
\gamma (\{1,2,3\})+\gamma (\{2,3,4\})= \max \{w_{12},w_{23}\}+\max \{w_{23},w_{34}\}.
\end{equation}
We claim that $w_{23}> \max\{w_{12},w_{34}\}$.
Assume to the contrary that $w_{23}\leq \max\{w_{12},w_{34}\}$.
We distinguish two cases and show that neither is possible.
If $w_{23}\leq \min\{w_{12},w_{34}\}$, then $w_{12}+w_{34} \geq w_{12}+w_{23}+w_{34}$ follows from \eqref{eq:4cover}, implying $w_{23}\leq 0$.
If $\min \{w_{12}, w_{34}\} < w_{23}\leq \max \{w_{12}, w_{34}\}$,
then $w_{23}+\max \{w_{12},w_{34}\} \geq w_{12}+w_{23}+w_{34}$ follows from \eqref{eq:4cover},
implying $\min \{w_{12}, w_{34}\}\leq 0$.
Neither case is possible.
Hence $w_{23}> \max\{w_{12},w_{34}\}$.
Then $2 w_{23} \geq w_{12}+w_{23}+w_{34}$ follows from \eqref{eq:4cover}, implying $w_{23}\geq w_{12}+w_{34}$.

\emph{\ref{itm:P3K3}}
Assume that $E(H_1)=\{12,23\}$, $E(H_2)=\{23,24,34\}$ and $w_{23}\geq w_{24}$.
Hence $\gamma (\{2,3,4\})=\max\{w_{23},w_{34}\}$.
Then \eqref{eq:P4} still applies, implying $w_{23}\geq w_{12}+w_{34}$.
And $w_{23}\geq w_{24}+w_{34}$ follows from Lemma \ref{thm:K3}.

\emph{\ref{itm:2K3}}
Assume that $E(H_1)=\{12,13,23\}$, $E(H_2)=\{23,24,34\}$ and $w_{12}\geq \max \{w_{13}, w_{23}\}$.
Thus
\begin{equation}
\gamma (\{1,2,3\})+\gamma (\{2,3,4\})=w_{12} + \max \{w_{23},w_{24}, w_{34}\}.
\end{equation}
Then $\max \{w_{23},w_{24}, w_{34}\}\geq w_{23}+w_{34}$ follows from \eqref{eq:4cover},
implying that $w_{24}\geq w_{23}+w_{34}$.
\end{proof}

\begin{lemma}[$P_4$-property]
\label{thm:P4}
Let $H\subseteq G$ with $E(H)=\{12,23,34\}$ be an induced subgraph.
If $\Gamma_G$ is population monotonic, then $w_{23} \geq w_{12}+w_{34}$.
\end{lemma}

\begin{proof}
Notice that $G[\{1,2,3\}]$ and $G[\{2,3,4\}]$ make a $2 P_3$-cover of $\{1,2,3,4\}$.
Then $w_{23} \geq w_{12}+w_{34}$ follows from Lemma \ref{thm:cover} \emph{\ref{itm:2P3}}.
\end{proof}

\begin{lemma}[Paw-property]
\label{thm:Paw}
Let $H\subseteq G$ with $E(H)=\{12,23,24,34\}$ be an induced subgraph with $w_{23}\geq w_{24}$.
If $\Gamma_G$ is population monotonic, then $w_{23}\geq w_{12}+w_{34}$ and $w_{23}\geq w_{24}+w_{34}$.
\end{lemma}
\begin{proof}
Notice that $G[\{1,2,3\}]$ and $G[\{2,3,4\}]$ make a $(P_3,K_3)$-cover of $\{1,2,3,4\}$.
Then $w_{23}\geq w_{12}+w_{34}$ and $w_{23}\geq w_{24}+w_{34}$ follow from Lemma \ref{thm:cover} \emph{\ref{itm:P3K3}}.
\end{proof}

Lemma \ref{thm:Paw} implies that for any $K_3$-subgraph $H$ in a graph inducing population monotonic matching games, endpoints of the maximum weight edge in $H$ are the only possible vertices in $H$ incident to other edges in the graph.

\begin{lemma}[Diamond-property]
\label{thm:Diamond}
Let $H\subseteq G$ with $E(H)=\{12,13,23,24,34\}$ be an induced subgraph.
If $\Gamma_G$ is population monotonic, then $w_{23}\geq w_{12}+w_{13}$ and $w_{23}\geq w_{24}+ w_{34}$.
\end{lemma}
\begin{proof}
On one hand, $G[\{1,2,4\}]$ and $G[\{2,3,4\}]$ make a $(P_3,K_3)$-cover of $\{1,2,3,4\}$.
Lemma \ref{thm:Paw} implies $w_{34} < \max\{w_{23}, w_{24}\}$.
On the other hand, $G[\{1,3,4\}]$ and $G[\{2,3,4\}]$ make another $(P_3,K_3)$-cover of $\{1,2,3,4\}$.
Lemma \ref{thm:Paw} implies $w_{24} < \max\{w_{23}, w_{34}\}$.
It follows that $w_{23}>\max\{w_{24},w_{34}\}$.
By symmetry, $w_{23}>\max\{w_{12},w_{13}\}$.
Then $w_{23}\geq w_{12}+w_{13}$ and $w_{23}\geq w_{24}+w_{34}$ follow from Lemma \ref{thm:K3}.
\end{proof}

As we shall see in Lemma \ref{thm:K4}, any population monotonic matching game has no $K_4$-subgraph.
Thus Lemma \ref{thm:Diamond} implies that for any two $K_3$-subgraphs $H_1$ and $H_2$ in a graph inducing population monotonic matching games, if $H_1$ and $H_2$ share a common edge $e$, then $e$ must be the maximum weight edge in both $H_1$ and $H_2$.

\subsection{Forbidden structures}
\label{subsec:obstructions}

This subsection develops some forbidden structures in population monotonic matching games.
Even though this list is not exhaustive, it is sufficient to achieve a complete characterization for population monotonic matching games.
For any two graphs $H_1$ and $H_2$, $H_1$ is said \emph{$H_2$-free} if $H_1$ has no induced subgraph isomorphic to $H_2$.

\begin{figure}[t]
\centering
\vspace{-1em}
\centerline{\includegraphics[width=1.45\textwidth]{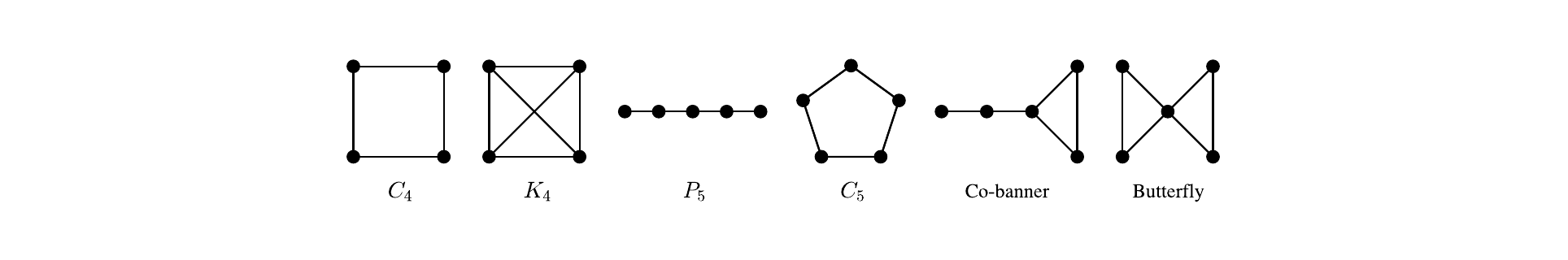}}
\vspace{-2em}
\caption{Some forbidden structures}
\label{fig:obstructions}
\vspace{-1em}
\end{figure}

Lemmas \ref{thm:C4} and \ref{thm:K4} concern two forbidden structures with $4$ vertices.

\begin{lemma}
\label{thm:C4}
If $\Gamma_{G}$ is population monotonic, then $G$ is $C_{4}$-free.
\end{lemma}
\begin{proof}
Assume to the contrary that $H\subseteq G$ is an induced subgraph isomorphic to $C_4$.
Without loss of generality, assume that $V(H)=\{1,2,3,4\}$ and $E(H)=\{12,23,34,14\}$.
Notice that $G[\{1,2,3\}]$ and $G[\{2,3,4\}]$ make a $2 P_3$-cover of $V(H)$,
and that $G[\{2,3,4\}]$ and $G[\{1,3,4\}]$ make another $2 P_3$-cover of $V(H)$.
Lemma \ref{thm:cover} \emph{\ref{itm:2P3}} implies $w_{23}\geq w_{12}+w_{34}$ and $w_{34}\geq w_{23}+w_{14}$.
A contradiction occurs.
\end{proof}

\begin{lemma}
\label{thm:K4}
If $\Gamma_{G}$ is population monotonic, then $G$ is $K_{4}$-free.
\end{lemma}

\begin{proof}
Assume to the contrary that $H\subseteq G$ is an induced subgraph isomorphic to $K_4$.
Without loss of generality, assume that $V(H)=\{1,2,3,4\}$.
By symmetry, further assume that $w_{13}=\max_{e\in E(H)}\{w_{e}\}$.
Lemma \ref{thm:K3} implies that $w_{13}\geq w_{12}+w_{23}$ and $w_{13}\geq w_{14}+w_{34}$.
Notice that $G[\{1,2,3\}]$ and $G[\{1,2,4\}]$ make a $2 K_3$-cover of $V(H)$,
and that $G[\{1,2,4\}]$ and $G[\{1,3,4\}]$ make another $2 K_3$-cover of $V(H)$.
Lemma \ref{thm:cover} \emph{\ref{itm:2K3}} implies that $w_{14}\geq w_{12}+w_{24}$ and $w_{12}\geq w_{14}+w_{24}$, which is absurd.
\end{proof}

Lemmas \ref{thm:P5C5} - \ref{thm:Butterfly} provide four forbidden structures with $5$ vertices.

\begin{lemma}
\label{thm:P5C5}
If $\Gamma_{G}$ is population monotonic, then $G$ is $(P_5, C_5)$-free.
\end{lemma}
\begin{proof}
Assume to the contrary that $H\subseteq G$ is an induced subgraph isomorphic to $P_5$ or $C_5$.
Without loss of generality, assume that $V(H)=\{1,2,3,4,5\}$ and $\{12,23,34,45\}\subseteq E(H)$.
Notice that both $G[\{1,2,3,4\}]$ and $G[\{2,3,4,5\}]$ are isomorphic to $P_4$.
Lemma \ref{thm:P4} implies that $w_{23}\geq w_{12}+w_{34}$ and $w_{34}\geq w_{23}+w_{45}$, which is absurd.
\end{proof}

A graph is a \emph{co-banner graph} if it is isomorphic to the fifth graph in Figure \ref{fig:obstructions}.

\begin{lemma}
\label{thm:CoBanner}
If $\Gamma_{G}$ is population monotonic, then $G$ is co-banner-free.
\end{lemma}
\begin{proof}
Assume to the contrary that $H\subseteq G$ is an induced subgraph isomorphic to co-banner.
Without loss of generality, assume that $V(H)=\{1,2,3,4,5\}$ and $E(H)=\{12,23,34,35,45\}$.
By symmetry, further assume that $w_{34}\geq w_{35}$.
Since $G[\{2,3,4,5\}]$ is isomorphic to paw, Lemma \ref{thm:Paw} implies that $w_{34}\geq w_{23}+w_{45}$.
Besides, $G[\{1,2,3,4\}]$ is isomorphic to $P_4$.
Lemma \ref{thm:P4} implies that $w_{23}\geq w_{12}+w_{34}$.
A contradiction occurs.
\end{proof}

A graph is a \emph{butterfly graph} if it is isomorphic to the sixth graph in Figure \ref{fig:obstructions}.

\begin{lemma}
\label{thm:Butterfly}
If $\Gamma_{G}$ is population monotonic, then $G$ is butterfly-free.
\end{lemma}
\begin{proof}
Assume to the contrary that $H\subseteq G$ is an induced subgraph isomorphic to butterfly.
Without loss of generality, assume that $V(H)=\{1,2,3,4,5\}$ and $E(H)=\{12,13,23,34,35,45\}$.
By symmetry, further assume that $w_{23}\geq w_{13}$ and $w_{34}\geq w_{35}$.
Notice that both $G[\{1,2,3,4\}]$ and $G[\{2,3,4,5\}]$ are isomorphic to paw.
Lemma \ref{thm:Paw} implies $w_{23}\geq w_{12}+w_{34}$ and $w_{34}\geq w_{23}+w_{45}$.
A contradiction occurs.
\end{proof}

Now we are ready to characterize population monotonic matching games.

\subsection{A complete characterization for population monotonicity}
\label{subsec:characterization}

A \emph{$k$-star} is a graph comprised of a clique $K$ of size $k$ and an independent set $I$,
where every vertex in $I$ is incident to at least one vertex in $K$.
We call vertices in $K$ \emph{centers} of the $k$-star.
In the following, we concentrate on $k$-stars with $k\leq 2$.
We refer to $1$-star as \emph{star} and to $2$-star as \emph{double-star}. 
Notice that a double-star degenerates to a star if all non-center vertices are connected to only one and the same center.

\begin{figure}[t]
\vspace{-1em}
\centerline{\includegraphics[width=1.45\textwidth]{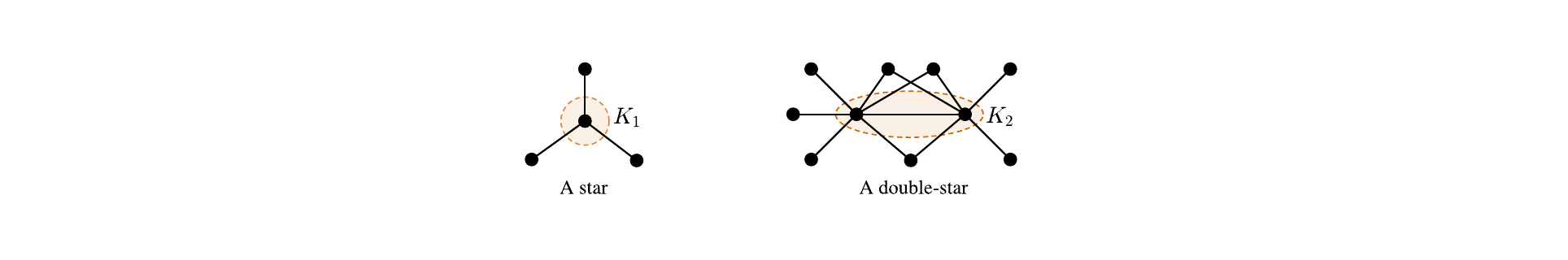}}
\vspace{-2em}
\caption{Examples for $k$-stars}
\label{fig:stars}
\vspace{-1em}
\end{figure}

We call adjacent vertices $u$ and $v$ in $G$ a \emph{dominant pair} if $w_{uv}\geq w_{uu'}+w_{vv'}$ for any pair of edges $uu'$ and $vv'$ incident to $uv$,
where vertices $u'$ and $v'$ might coincide.
In particular, we still call $u$ and $v$ a \emph{dominant pair} if $uv$ is the only edge incident to $v$ and $w_{uv}\geq w_{uu'}$ for any other edge $uu'$ incident to $u$.
Now we are ready to present our main result.

\begin{theorem}
\label{thm:main}
A matching game is population monotonic if and only if every component of the underlying graph is a double-star where the two centers make a dominant pair.
\end{theorem}
\begin{proof}
Let $\Gamma_G=(N,\gamma)$ be the matching game on $G=(V,E;w)$, where $G$ a simple graph with weight function $w:E\rightarrow \mathbb{R}_+$.
Without loss of generality, assume that $G$ is connected.

$(\Longrightarrow)$.
Assume that $\Gamma_G$ is population monotonic.
We distinguish two cases to show that $G$ is a double-star where the two centers make a dominant pair.

Case $1$: $G$ is a tree.
By Lemma \ref{thm:P5C5}, $G$ is a double-star where every non-center vertex is adjacent to precisely one center.
Let $u$ and $v$ be the two centers of $G$.
In particular, let $u$ and $v$ be the endpoints of a maximum weight edge in $G$ if $G$ is a star.
Lemma \ref{thm:P4} implies that $u$ and $v$ make a dominant pair in $G$. 

Case $2$: $G$ is not a tree.
Lemma \ref{thm:P5C5} suggests that $G$ contains $K_3$-subgraphs.
We call two $K_3$-subgraphs of $G$ \emph{disjoint} if they share no common vertex and \emph{non-disjoint} otherwise.

We first show that any two non-disjoint $K_3$-subgraphs of $G$ share a common edge.
Moreover, the common edge is the maximum weight edge in either $K_3$-subgraph.
Assume to the contrary that $H_1$ and $H_2$ are two non-disjoint $K_3$-subgraphs without common edges.
Without loss of generality, assume that $V(H_1)=\{1,2,3\}$ and $V(H_2)=\{3,4,5\}$.
Clearly, $H_1 \cup H_2$ is a butterfly graph.
Lemma \ref{thm:Butterfly} implies that there exist crossing edges in $G$ between vertex sets $\{1,2\}$ and $\{4,5\}$.
Lemmas \ref{thm:C4} and \ref{thm:K4} suggest that there exists at most one crossing edge in $G$ between $\{1,2\}$ and $\{4,5\}$.
We may assume that $14$ is the only crossing edge between $\{1,2\}$ and $\{4,5\}$.
Notice that both $G[\{1,2,3,4\}]$ and $G[\{1,3,4,5\}]$ are isomorphic to diamond.
Lemma \ref{thm:Diamond} implies that $w_{13}\geq w_{14}+w_{34}$ and $w_{34}\geq w_{13}+w_{14}$, which is absurd.
Hence any two non-disjoint $K_3$-subgraphs of $G$ share a common edge.
Moreover, Lemma \ref{thm:Diamond} implies that the common edge is the maximum weight edge in either $K_3$-subgraph.

Next we show that all $K_3$-subgraphs of $G$ are pairwise non-disjoint.
Moreover, all $K_3$-subgraphs of $G$ share a common edge which is the maximum weight edge in every $K_3$-subgraph.
Assume to the contrary that there exist disjoint $K_3$-subgraphs in $G$.
Then Lemma \ref{thm:P5C5} guarantees that there are disjoint $K_3$-subgraphs $H_1$ and $H_2$ with crossing edges between $V(H_1)$ and $V(H_2)$.
Moreover, Lemmas \ref{thm:C4} and \ref{thm:Butterfly} suggest that there exists at most one crossing edge between $V(H_1)$ and $V(H_2)$.
Without loss generality, assume that $V(H_1)=\{1,2,3\}$, $V(H_2)=\{4,5,6\}$ and $14$ is the only crossing edge between $V(H_1)$ and $V(H_2)$.
However, $G[\{1,2,3,4,5\}]$ is isomorphic to co-banner, which contradicts Lemma \ref{thm:CoBanner}.
Hence all $K_3$-subgraphs of $G$ are pairwise non-disjoint.
Recall that any two non-disjoint $K_3$-subgraphs share a common edge which is the maximum weight edge in both of them.
Moreover, Lemma \ref{thm:K3} implies that the maximum weight edge is unique in every $K_3$-subgraph.
Hence all $K_3$-subgraphs of $G$ share a common edge which is the maximum weight edge in every $K_3$-subgraph.

Finally, we show that $G$ is a double-star where the two centers are the endpoints of the common edge for all $K_3$-subgraphs in $G$.
Moreover, the two centers of $G$ make a dominant pair.
Let $uv$ be the common edge of all $K_3$-subgraphs in $G$.
In particular, if there is only one $K_3$-subgraph in $G$, let $uv$ be the edge with maximum weight in the $K_3$-subgraph.
Thus $uv$ is the maximum weight edge in every $K_3$-subgraph of $G$.
Lemmas \ref{thm:Paw} and \ref{thm:CoBanner} suggest that vertices outside $K_3$-subgraphs of $G$ are adjacent to either $u$ or $v$.
It follows that $G$ is a double-star with centers $u$ and $v$.
Moreover, Lemmas \ref{thm:P4} - \ref{thm:Diamond} imply that $u$ and $v$ make a dominant pair.

$(\Longleftarrow)$.
Assume that $G$ is a double-star where the two centers make a dominant pair.
Let $u$ and $v$ be the two centers of $G$.
In particular, let $u$ and $v$ be the endpoints of a maximum weight edge if $G$ is a star.
We prove the matching game $\Gamma_G$ on $G$ is population monotonic by constructing a PMAS $\mathcal{X}=(\boldsymbol{x}_S)_{S\in 2^N \backslash \{\emptyset\}}$.
Let $S$ be a nonempty subset of $N$.
Define $\boldsymbol{x}_{S}=(x_{S,i})_{i\in S}$ as follows.
For any $i\in S\backslash \{u,v\}$, let $x_{S,i}=0$.
For values of $x_{S,u}$ and $x_{S,v}$, we distinguish two cases.
$(1)$ If $S\cap \{u,v\}=\{i\}$, then $x_{S,i}=\sigma_{i} (S)$,
where $\sigma_{i}(S)$ denotes the maximum weight of edges incident to $i$ in $G[S]-uv$.
$(2)$ If $S\cap \{u,v\} = \{u,v\}$,
then let $x_{S,u}=\frac{\sigma_u}{\sigma_u + \sigma_v} w_{uv}$ and $x_{S,v}=\frac{\sigma_v}{\sigma_u + \sigma_v} w_{uv}$,
where $\sigma_{u}$ and $\sigma_{v}$ denote the maximum weight of edges incident to $u$ and $v$ in $G-uv$, respectively.
It remains to show that $\mathcal{X}=(\boldsymbol{x}_S)_{S\in 2^N\backslash \{\emptyset\}}$ defined above is indeed a PMAS for $\Gamma_G$.

We first consider the efficiency condition.
Let $S$ be a nonempty subset of $N$.
Since $G$ is a double-star with $u$ and $v$ being a dominant pair,
we have
\begin{equation*}
\gamma (S)=
\begin{cases}
~w_{uv} & \text{if $u,v\in S$,}
\vspace{1.5 mm}\\
\sigma_u (S) & \text{if $u\in S$, $v\not\in S$,}
\vspace{1.5 mm}\\
\sigma_v (S) & \text{if $u\not\in S$, $v\in S$,}
\vspace{1.5 mm}\\
\quad 0 & \text{if $u,v\not\in S$.}
\end{cases}
\end{equation*}
Hence $\sum_{i\in S} x_{S,i}=\gamma (S)$ holds in any case.

Now we check the monotonicity condition.
Let $S$ and $T$ be two subsets of $N$ with $S\subseteq T$.
It suffices to show that $x_{S,i}\leq x_{T,i}$ for any $i\in S\cap \{u,v\}$.
Without loss of generality, assume that $u\in S$.
We proceed by distinguishing three cases.
$(1)$ If $v\in S$, then $x_{S,u}=\frac{\sigma_{u}}{\sigma_{u} + \sigma_{v}} w_{uv} =x_{T,u}$.
$(2)$ If $v\in T\backslash S$, then $x_{S,u}=\sigma_u (S)\leq \sigma_u \leq \frac{\sigma_u}{\sigma_u +\sigma_v} w_{uv}=x_{T,u}$, since $w_{uv}\geq \sigma_u+\sigma_v$.
$(3)$ If $v\not\in T$, then $x_{S,u}=\sigma_u (S)\leq \sigma_u (T)=x_{T,u}$.
Hence $x_{S,u}\leq x_{T,u}$ holds in any case.
By symmetry, we have $x_{S,i}\leq x_{T,i}$ for any $i\in S\cap \{u,v\}$.

Therefore, the matching game $\Gamma_G$ on $G$ is population monotonic.
\end{proof}

Graphs that are double-stars can be determined efficiently.
Indeed, the centers in a double-star are the only possible vertices with degree larger than $2$,
and all non-center vertices make an independent set.
Besides, dominant pairs in double-stars can be verified efficiently.
Hence we have the following corollary for our main result.

\begin{corollary}
The population monotonicity of a matching game can be determined efficiently.
\end{corollary}

\section{Discussion}
\label{sec:disscussion}
In this paper, we study matching games and provide a necessary and sufficient characterization for the population monotonicity.
Prior to our work, Immorlica et al. \cite{IMM08} studied the limitation of approximate PMASes and proved that no constant approximation factor exists even for simple matching games.
Hence our result completes the line of research on PMASes for matching games.

One possible working direction for matching games is to study the existence of allocation schemes consisting of core allocations for every coalition,
where Deng et al. \cite{DINZ00} have achieved a sufficient characterization for simple matching games.
Our result might be helpful since a PMAS provides a core allocation for every coalition in a population monotonic way.
Besides, variants of matching games introduced by Kumabe and Maehara \cite{KM20-IJCAI, KM20-AAMAS} are also worth studying.


\bibliographystyle{habbrv}
\bibliography{reference}
\end{document}